\documentclass[lettersize,journal]{IEEEtran}
\usepackage[table]{xcolor}
\usepackage{cite}
\usepackage{amsmath,amssymb,amsfonts,amsthm}
\usepackage{textcomp}
\usepackage{aligned-overset}

\usepackage{booktabs}
\usepackage{graphicx}
\usepackage{siunitx} 
\usepackage{gensymb}
\usepackage{algorithm,algorithmic}
\usepackage{enumitem}
\usepackage{atbegshi}
\usepackage{tikz}
\usetikzlibrary{arrows}
\usetikzlibrary{shapes}
\usetikzlibrary{decorations}
\usetikzlibrary{calc}
\usetikzlibrary{arrows}
\usetikzlibrary{shapes}
\usetikzlibrary{calc}
\usetikzlibrary{tikzmark,fit}
\usetikzlibrary{spy}
\usepackage{textcomp}

\usepackage{pgfplots}

\makeatletter
\let\MYcaption\@makecaption
\makeatother

\usepackage[font=footnotesize]{subcaption}

\makeatletter
\let\@makecaption\MYcaption
\makeatother

\makeatletter
\long\def\@makecaption#1#2{%
  \ifx\@captype\@IEEEtablestring%
    \@IEEEtablecaptionsepspace
    \setbox\@tempboxa\hbox{\normalfont\footnotesize #1.\nobreakspace\nobreakspace #2}%
    \ifdim \wd\@tempboxa >\hsize%
      \setbox\@tempboxa\hbox{\normalfont\footnotesize #1.\nobreakspace\nobreakspace}%
      \parbox[t]{\hsize}{\normalfont\footnotesize \noindent\unhbox\@tempboxa #2}%
    \else%
      \hbox to\hsize{\normalfont\footnotesize \box\@tempboxa\hfil}%
    \fi%
  \else
    \@IEEEfigurecaptionsepspace
    \setbox\@tempboxa\hbox{\normalfont\footnotesize #1.\nobreakspace\nobreakspace #2}%
    \ifdim \wd\@tempboxa >\hsize%
      \setbox\@tempboxa\hbox{\normalfont\footnotesize #1.\nobreakspace\nobreakspace}%
      \parbox[t]{\hsize}{\normalfont\footnotesize \noindent\unhbox\@tempboxa #2}%
    \else%
      \hbox to\hsize{\normalfont\footnotesize \box\@tempboxa\hfil}%
    \fi%
  \fi}
\makeatother

\pgfplotsset{compat=1.10}
\usetikzlibrary{pgfplots.statistics, pgfplots.colorbrewer, pgfplots.groupplots}
\usepackage{pgfplotstable}

\usepackage{textcomp}

\usepackage{array}
\usepackage{multirow}
\usepackage{tablefootnote}

\newtheorem{theorem}{Theorem}
\newtheorem{lemma}[theorem]{Lemma}
\newtheorem{assumption}{Assumption}
\newtheorem{remark}{Remark}

\usepackage[hidelinks]{hyperref}

\newcommand{\norm}[1]{\left\lVert #1 \right\rVert}
\newcolumntype{P}[1]{>{\centering\arraybackslash}p{#1}}
\newcolumntype{s}[1]{>{\columncolor{lightgray!20}\centering\arraybackslash}p{#1}}

\usepackage{xspace}

\newcommand{\ie}{i\/.\/e\/.\/,~}
\newcommand{\eg}{e\/.\/g\/.\/,~}

\newcommand{\fig}{Fig\/.\/~}

\newcommand{\tabe}{Table~}

\newcommand{\sect}{Sec\/.\/~}

\newcommand{\algo}{Algorithm~}
\newcommand{\alg}{Alg\/.\/~}
\newcommand{\theo}{Theorem~}
\newcommand{\lem}{Lemma~}

\newif\ifcommentjohannes
\commentjohannestrue

\def\BibTeX{{\rm B\kern-.05em{\sc i\kern-.025em b}\kern-.08em
    T\kern-.1667em\lower.7ex\hbox{E}\kern-.125emX}}

\newif\ifshowchanges
\showchangesfalse

\newcommand{\rev}[1]{%
  \ifshowchanges
    {\color{blue}#1}%
  \else
    #1%
  \fi
}

\begin{document}
\title{Approximate non-linear model predictive control \\with safety-augmented neural networks}
\author{Henrik Hose, Johannes K{\"o}hler, Melanie N. Zeilinger, and Sebastian Trimpe
\thanks{This work is funded in part by the DFG RTG 2236 “UnRAVeL” and by the Swiss National Science Foundation under NCCR Automation (grant agreement 51NF40 180545).
Simulations were performed with computing resources granted by RWTH Aachen University under project RWTH1288.}
\thanks{H.~Hose and S.~Trimpe are with the Institute for Data Science in Mechanical Engineering, RWTH Aachen University, Germany (e-mail: henrik.hose@dsme.rwth-aachen.de; trimpe@dsme.rwth-aachen.de).}
\thanks{J. K{\"o}hler and M.~N. Zeilinger are with the Institute for Dynamic Systems and Control, ETH Z{\"u}rich, Switzerland (e-mail: jkoehle@ethz.ch; mzeilinger@ethz.ch).}}

\maketitle

\IEEEpubid{
\begin{minipage}{\textwidth}\ \\[48pt]
\scriptsize
Accepted final version. Accepted for publication in: \textit{IEEE Transactions on Control Systems Technology}, 2025. DOI: 10.1109/TCST.2025.3590268\\
© 2025 IEEE.
Personal use of this material is permitted.
Permission from IEEE must be obtained for all other uses, in any current or future media, including reprinting/republishing this material for advertising or promotional purposes, creating new collective works, for resale or redistribution to servers or lists, or reuse of any copyrighted component of this work in other works.
\end{minipage}
}

\begin{abstract}
    Model predictive control (MPC) achieves stability and constraint satisfaction for general nonlinear systems, but requires computationally expensive online optimization.
    This paper studies approximations of such~MPC controllers via neural networks~(NNs) to achieve fast online evaluation.
    We propose safety augmentation that yields deterministic guarantees for convergence and constraint satisfaction despite approximation inaccuracies.
    We approximate the entire input sequence of the~MPC with NNs, which allows us to verify online if it is a feasible solution to the MPC problem.
    We replace the NN solution by a safe candidate based on standard MPC techniques whenever it is infeasible or has worse cost.
    Our method requires a single evaluation of the NN and forward integration of the input sequence online, which is fast to compute on resource-constrained systems\rev{, typically within 0.2ms}.
    The proposed control framework is illustrated using \rev{three} numerical non-linear~MPC benchmarks of different complexity, demonstrating computational speedups that are orders of magnitude higher than online optimization.
    In the examples, we achieve deterministic safety through the safety-augmented NNs, where a naive NN implementation fails.
\end{abstract}

\begin{IEEEkeywords}
Non-linear model predictive control; approximate~MPC; machine learning; constrained control
\end{IEEEkeywords}

\section{Introduction}
\IEEEPARstart{M}{odel} predictive control~(MPC) is an advanced control strategy for non-linear systems that provides theoretical guarantees for constraint satisfaction and stability~\cite{rawlings2017model}.
Convenient software packages are available to implement the required constrained optimization problems and parse them for dedicated numerical solvers, \eg\cite{verschueren2022acados}.
However, practical applications remain limited by long computing times in resource-constrained systems and, for general non-linear systems, it is difficult to acquire a priori certification of computation time bounds.
Neural networks (NNs) have been successfully trained to imitate (approximate) the behavior of MPC schemes~\cite{parisini1998nonlinear,aakesson2006neural,hertneck2018learning,nubert2020safe,karg2021probabilistic,klauvco2019machine,vaupel2020accelerating, zhang2020near,wang2023policy,chen2018approximating,paulson2020approximate,drgona2020learning, drgovna2022learning, cao2020deep, li2022using, tabas2022safe, carius2020mpc, reske2021imitation}, resulting in fast online sampling suitable for embedded applications\rev{, see \cite{gonzalez2023neural} for a recent survey that highlights the numerical speed-ups of such approximate MPCs for different problems.}
Yet, providing hard guarantees for constraint satisfaction and convergence for general, non-linear systems remains challenging for such approximate~MPCs.
\begin{figure}[!ht]
\center
\begin{tikzpicture}

    \draw[draw=black,fill=lightgray!20,thick, name=AUG] (-0.5,1.1) rectangle ++(5.9,-3.8);
    \node[align=center, font={\small}] at (2.4,0.8){\textbf{Safety-augmented NN}};
    
    \node [rectangle, thick, draw, align=center, fill=white, minimum height = 1cm, text width=5cm, font={\small}] (NN) at (2.4,0) 
    {\textbf{Approximate MPC} (NN)\\infer input sequence $\Pi_\text{NN}(x)\in\mathcal{U}^N$};
    
    \node [rectangle, thick, draw, fill=white, align=center, minimum height = 1cm, text width=2cm, font={\small}] (SYS) at (-1.8,-0.65) 
    {\textbf{System}\\$x^+=f(x,u)$};
    
    \node[align=left, font={\small}] at (2.85,-2.35){Candidate $\{\pmb{u}_{1:N-1},K_\text{f}(\cdot)\}\in\mathcal{U}^N(x)$};

    \node [rectangle,
        thick,
        draw,
        align=center,
        minimum height = 1cm,
        text width=5cm,
        fill=white,
        font={\small}] (SOEB) at (2.4,-1.4) 
        {\textbf{Safe online evaluation}\\
        if~$\Pi_\text{NN}$ is safe and stable, $\pmb{u}=\Pi_{\text{NN}}$, \\
        else use candidate as fallback};

    \begin{scope}[
        every node/.style={rectangle, align=center, font={\scriptsize}},
        every edge/.style={thick, draw},
        every path/.style={thick, draw}
    ]
        \draw [->] (SOEB.west) -| (SYS.south);
        \node at (-1.0,0.15) {$x$};
        \node at (-1.0,-1.6) {$\pmb{u}_0$};
        \draw [<-] (SOEB.north) -- (NN.south);
        \draw [->] (SYS.north) |- (NN.west);

        \draw [->] (SOEB.-165) to [loop, out=-65,in=-115, min distance=0.8cm](SOEB.-166.8);

    \end{scope}

\end{tikzpicture}
\caption{Approximate MPC with safety-augmented NN. Input sequences from a NN controller are checked by the proposed safe online evaluation (see \alg\ref{alg:online-validation}). If they yield constraint satisfaction and stability, they are applied to the system. Otherwise, a safe candidate obtained from the last time step -- shifted with appended terminal controller -- is used as fallback.}\label{fig:mainfig}
\vspace{-1em}
\end{figure}

\rev{\subsection{Related work}\label{sec:related-work}}
\rev{Explicit solutions to~MPC problems can be obtained exactly for medium-sized linear systems~\cite{bemporad2002explicit} and approximately for non-linear and large linear systems (e.g., with NNs~\cite{parisini1998nonlinear,aakesson2006neural} or multi-parametric programming~\cite{johansen2004approximate}). 
Due to the approximation error, the resulting approximate MPC does in general not guarantee constraint satisfaction, or guarantees are only valid under sufficiently small approximation errors, that are difficult to verify~\cite{chakrabarty2016support,canale2009set}.
In the following, we discuss existing methods to ensure constraint satisfaction with approximate MPC and classical results on suboptimal~MPC similar to the proposed method. }

\textbf{Robust with respect to approximation:}
One way to account for inevitable approximation errors is to explicitly add robustness to the underlying MPC scheme.
\rev{In~\cite{pin2013approximate}, constraint satisfaction and stability are shown with Lipschitz-based constraint tightening~(see~\cite{manzano2020robust,marruedo2002input}), but the required admissible approximation error is typically too small to be achieved in practice.
In~\cite{hertneck2018learning, nubert2020safe}, more advanced robustification (see~\cite{kohler2020computationally, houska2019robust}) is used to handle realistic approximation errors, for which bounds are obtained via statistical validation.
However, the required robust MPC introduces conservatism and a uniform bound on the approximation error can be hard to achieve in practice.}
A similar statistical approach is used in~\cite{karg2021probabilistic} to validate safety.
Due to the probabilistic nature of these approaches, a non-zero probability of faulty inputs remains, and safety is not deterministically guaranteed in closed loop.

\textbf{NN as warm start:}
Using NN predictions as warm starting solutions to numerical solvers potentially reduces computation time, while the solver guarantees constraint satisfaction, meaning that standard MPC guarantees apply.
A speedup by a factor of two is indeed observed for linear systems with active set solvers in~\cite{klauvco2019machine} and for large-scale systems in~\cite{chen2022large}.
However, this method still requires online optimization, for which worst case computation time bounds can be long.
\rev{Furthermore, in~\cite{vaupel2020accelerating}, warm starting with NN predictions was observed to increase computation time in non-linear numerical examples.}

\textbf{Safety filters:}
Safety enforcing methods aim to modify the outputs of the NN by finding the closest safe inputs.
\rev{For linear systems, one can consider an \textit{explicit} safety filter
 to ensure  safety~\cite{karg2020efficient,chen2018approximating,paulson2020approximate}, however, application for general, non-linear systems is challenging 
due to required invariant sets~\cite{wabersich2023data}.}
\rev{Similar to non-linear \textit{predictive} safety filters~\cite{wabersich2021predictive},
~\cite{wang2023policy} proposes approximating a complete sequence of inputs of a non-linear~MPC with an~NN and solving an optimization problem online that finds the closest safe trajectory.
Unfortunately, solving online optimization problems can be computationally expensive and does not admit a worst-case computation time bound.
While such predictive safety filters can also be approximated with NNs~\cite{didier2023approximate}, safety guarantees again require uniform approximation error bounds.
Unlike warm starts or safety filters, our method eliminates online optimization altogether, significantly reducing computation time.}\looseness-1

\textbf{Constrained learning:}
\rev{%
Constrained learning methods aim to improve performance and constraint satisfaction when imitating an MPC, 
by considering cost and constraints of the MPC more directly in the training loss~\cite{drgona2020learning, drgovna2022learning,cao2020deep}. %
However, none of these works provides constraint satisfaction guarantees for states not seen during training.
In~\cite{zhang2020near}, constrained learning is used to approximate primal and dual~MPC problems. In the control loop, primal feasibility (safety) and suboptimality are checked, which is also part of our method.
However, results in~\cite{zhang2020near} are valid for linear systems.}
By contrast, our method is applicable to large scale systems and provides deterministic guarantees without the need for a separate backup controller.

\textbf{Suboptimal MPC:}
In non-linear MPC, it is not always necessary to solve optimization problems to optimality.
\rev{As long as the optimized input sequence is feasible and reduces the value function compared to the candidate, the closed-loop system will be persistently feasible and stable~\cite{scokaert1999suboptimal}, which can be enforced in robust suboptimal~MPC by a constraint on cost reduction over the previous time step~\cite{zeilinger2014real} and has some inherent robustness when warm-stared~\cite{allan2017inherent}.
While the setting of suboptimal MPC is quite different to ours, we use a similar idea: We ensure persistent feasibility and convergence by providing a feasible, suboptimal control sequence from an NN that performs at least as well as the candidate sequence.}

\subsection{Contribution}
We propose approximate~MPC with safety augmented NNs (illustrated in \fig\ref{fig:mainfig}), a novel algorithm that provides deterministic guarantees for feasibility and convergence.
In particular, we approximate the optimal open-loop input sequence of a MPC with a function approximator that is simple to evaluate (here, an~NN).
This allows us to check online whether the approximator provides a feasible and stabilizing solution to the~MPC problem, and is thus safe. 
This approach requires a single forward simulation of the system dynamics, a small additional computational cost, but significantly faster than MPC with online optimization.
We use standard MPC techniques online to generate a safe candidate input sequence based on a previous safe input sequence with appended terminal control law.
If the~NN solution is unsafe or does not improve the cost compared to the candidate, the candidate is applied.
This approach ensures constraint satisfaction (safety) and convergence without any assumptions on the approximation quality.

If the~NN approximation is poor, the safety augmentation would often \emph{not} use the NN.
In an additional result, we therefore show that the naive~NN would be safe if a robust MPC is sufficiently well approximated.
We use this in a constructive way by approximating robust MPC schemes.

Numerical examples on \rev{three} non-linear~MPC benchmark systems illustrate that our approach: (i) reduces the overall computational time by one to four orders of magnitude compared to online optimization; and (ii) provides deterministic safety.
These examples also demonstrate the importance of the proposed safety augmentation, as naively (always) applying the~NN solution leads to non-negligible constraint violations.

\subsection{Notation}
The set of integers is denoted by~$\mathbb{I}$, the set of integers from~$0$ to~$N$ by~$\mathbb{I}_{0:N}$, and the set of non-negative integers by~$\mathbb{N}$.
By~$\mathcal{K}_\infty$, we denote the class of functions~$\alpha:\mathbb{R}_{\geq0} \rightarrow \mathbb{R}_{\geq0}$ which are continuous, strictly increasing, satisfy~$\alpha(0)=0$, and are unbound.
For two sets,~$\mathcal{A}$ and~$\mathcal{B}$,~$\mathcal{A}\oplus\mathcal{B}$ is the Minkowski sum and~$\mathcal{A}\ominus\mathcal{B}$ is the Pontryagin difference.
We denote the square diagonal matrix with elements~$v\in\mathbb{R}^n$ on the diagonal as~$\text{diag}(v)\in\mathbb{R}^{n\times n}$.~$I_N$ is the identity matrix of size~$N$.

\section{System description}\label{sec:mpc}
We consider general, non-linear, discrete time systems 
\begin{align}
    x(t+1) = f(x(t),u(t)) \label{eqn:transition}
\end{align}
with state~$x(t)\in\mathbb{R}^{n_\text{x}}$, input~$u(t)\in\mathbb{R}^{n_\text{u}}$, and~$t\in\mathbb{N}$.
We assume~$f$ is continuous and consider compact constraint sets for states~$\mathcal{X}$ and inputs~$\mathcal{U}$.
The objective is to drive the system to \rev{an equilibrium (without loss of generality we use~$x=0$ with~$f(0,0)=0$)} and ensure constraint satisfaction
\begin{align}\label{eqn:constraints}
(x(t),u(t))\in\mathcal{X}\times\mathcal{U}\;\forall\;t\in\mathbb{N}
\end{align} for a large set of initial conditions.
This objective can be achieved by MPC~\cite{rawlings2017model}.

\subsection{MPC problem}\label{sec:mpc:general}
In MPC, a sequence of inputs~$\pmb{u}~=~[u_0^\top,~\dots,~u_{N-1}^\top]^\top\in\mathcal{U}^N$ is computed.
The predicted state when applying~(\ref{eqn:transition}) for~$k$ inputs from the sequence~$\pmb{u}$, starting at~$x$, is~$\phi(k;x;\pmb{u})$.
In MPC, we consider a finite-horizon cost of the form
\begin{align}\label{eqn:nominalcost}
    V(x,\pmb{u})=V_\text{f}(\phi(N;x;\pmb{u}))+\sum_{k=0}^{N-1} \ell(\phi(k;x;\pmb{u}),\rev{u_k}),
\end{align}
with horizon~$N$, non-negative stage cost~$\ell(x,u)$, and non-negative terminal cost~$V_\text{f}(x)$. 
We assume \rev{that $V_\text{f}$ and $\ell$ are continuous, which implies continuity of $V$, as $f$ is continuous.}
We denote the set of feasible input trajectories from a given state~$x$ by
\begin{align}\label{eqn:feasibletrajectories}
\begin{split}
    \mathcal{U}^N(x) = \{ \pmb{u}\in\mathcal{U}^N \mid &\phi(k;x;\pmb{u})\in\mathcal{X}, k\in\mathbb{I}_{0:N-1},\\
    &\phi(N;x;\pmb{u})\in\mathcal{X}_\text{f} \},
\end{split}
\end{align}
where~$\mathcal{X}_\text{f}\subseteq\mathcal{X}$ is the terminal set with~$0\in\text{int}(\mathcal{X}_\text{f})$.
The corresponding nominal MPC formulation is
\begin{align}\label{eqn:nominalmpc}
    V^*(x) =&\min_{\pmb{u}\in\mathcal{U}^N(x)} V(x,\pmb{u}).
\end{align}
The feasible set of~(\ref{eqn:nominalmpc}), \ie the set where~(\ref{eqn:nominalmpc}) has a solution, is denoted by~$\mathcal{X}_\text{feas}$.
Solving this optimization problem for the current value~$x$ gives a corresponding trajectory of optimal inputs~$\pmb{u}$, which we denote by~$\Pi_{\text{MPC}}:\rev{\mathcal{X}} \rightarrow \mathcal{U}^N$.
Applying the first input of this sequence,~$\pmb{u}_0$, to the system yields a control policy of the form~$u=\kappa_\text{MPC}(x)$, which can be applied in closed loop to the system.
We consider the following standard conditions for MPC~\cite{rawlings2017model}:
\begin{assumption}\label{ass:mpcingredients}
There exist a function~$\alpha_\ell\in\mathcal{K}_\infty$ and a \rev{terminal feedback controller}~$K_\text{\normalfont f}:\mathcal{X}_\text{\normalfont f}\rightarrow\mathcal{U}$
such that for any~$x\in\mathcal{X}$,~$u\in\mathcal{U}$:
\begin{itemize}[noitemsep,topsep=0pt,parsep=0pt,partopsep=0pt,leftmargin=10pt]
    \item $\ell(x,u)\geq\alpha_\ell(\norm{x})$,
    \item the terminal set is positively invariant under the terminal controller $x\in\mathcal{X}_\text{\normalfont f} \implies f(x,K_\text{\normalfont f}(x))\in\mathcal{X}_\text{\normalfont f}$,
    \item the terminal cost is a control Lyapunov function for~$x\in\mathcal{X}_\text{\normalfont f}$ \rev{under the terminal feedback control $K_\text{\normalfont f}$, that is} $V_\text{\normalfont f}(f(x,K_\text{\normalfont f}(x)))-V_\text{\normalfont f}(x) \leq -\ell(x,\rev{K_\text{\normalfont f}(x)})$.
\end{itemize}
\end{assumption}
These conditions ensure stability of the closed loop system and persistent feasibility, \ie if a solution to~(\ref{eqn:nominalmpc}) exists at~$t=0$, then it exists for all times~$t\in\mathbb{N}$.
Corresponding proofs alongside constructive methods to determine~$\mathcal{X}_\text{f}$,~$K_\text{f}$,~and~$V_\text{f}$ can be found in~\cite{rawlings2017model}.
Implementing MPC requires online optimization to solve~(\ref{eqn:nominalmpc}), which can be computationally expensive.
We circumvent this issue by approximating the MPC with NNs.

\subsection{Problem formulation: approximate MPC with guarantees}\label{sec:mpc:approximate}
We approximate the MPC via an NN \rev{trained} offline by first generating a dataset of initial conditions and corresponding, optimal input sequences, \ie solutions to~(\ref{eqn:nominalmpc}), and then training the NN in a standard supervised learning fashion~\cite{hertneck2018learning,aakesson2006neural,nubert2020safe,karg2021probabilistic,klauvco2019machine,vaupel2020accelerating}.
We do not make any further assumptions on NN structure or learning procedure (\eg loss function).
The approximation~$\Pi_\text{NN}\approx\Pi_\text{MPC}$ maps a state~$x$ to a sequence of inputs~$\pmb{u}$ with horizon length~$N$ as
$\Pi_\text{NN}: \rev{\mathcal{X}} \rightarrow \mathcal{U}^N$.
The approximation is not exact, ~$\Pi_\text{NN}\neq\Pi_\text{MPC}$. Therefore, naively (always) applying the first input~$\pmb{u}_0$ of~$\pmb{u}(t)=\Pi_\text{NN}(x(t))$ to the system does not guarantee constraint satisfaction; that is, $x(t)\notin\mathcal{X}$ for some~$t\in\mathbb{N}$ when applying $\pmb{u}(t)=\Pi_\text{NN}(x(t))$ and~$x(0)\in\mathcal{X}_\text{feas}$ (see Sec.~\ref{sec:numericalexamples} for examples where naively applying~$\Pi_\text{NN}$ fails).

\textit{Problem statement:}
We seek to develop a safety augmentation such that for a given NN approximation,~$\Pi_\text{NN}$, of an MPC~(\ref{eqn:nominalmpc}), the system is guaranteed to satisfy constraints and asymptotically converge; that is,~$x(t)\in\mathcal{X}$ for all~$t\in\mathbb{N}$ and~$x(t)\rightarrow0$ for~$t\rightarrow\infty$.

\section{Safe online evaluation}\label{sec:online-validation}
In this section, we address the problem stated above by augmenting the NN in a safe online evaluation.
We check online if $\Pi_\text{NN}(x(t))$ is a feasible solution to~(\ref{eqn:nominalmpc}) and provide a deterministically safe feedback policy that chooses between the NN sequence, if it is safe, and a safe fallback candidate.
The safety augmentation further ensures that the closed-loop system converges.
This idea is implemented in \algo\ref{alg:online-validation}.
\begin{algorithm}[H]
\caption{Approximate MPC with safety-augmented NN}\label{alg:online-validation}
 \textbf{Input}: $\Pi_\text{NN}(\cdot)$, $\mathcal{U}^N(\cdot)$, $K_\text{f}(\cdot)$, $V(\cdot,\cdot)$, $\tilde{\pmb{u}}(0)=\pmb{u}_{\text{init}}$ \\
 \vspace{-1.2em}
\begin{algorithmic}[1]
\FORALL{ $t\in\mathbb{N}$ }
    \STATE $x \gets x(t)$ \COMMENT{measure state at time~$t$}\label{alg:online-validation:newstate}
    \STATE $\hat{\pmb{u}}(t) \gets \Pi_\text{NN}(x)$ \COMMENT{evaluate approximation}\label{alg:online-validation:evaluateNN}
    \IF{$\hat{\pmb{u}}(t) \in \mathcal{U}^N(x)$ \COMMENT{check if~$\hat{\pmb{u}}(t)$ is feasible}} \label{alg:online-validation:feasible}
        \STATE $\pmb{u}(t)\gets\mathrm{argmin}_{\pmb{u}(t)\in\{\tilde{\pmb{u}}(t),\hat{\pmb{u}}(t)\}} V(x,\pmb{u}(t))$ \COMMENT{choose min cost}\label{alg:cost-decrease}
    \ELSE
        \STATE $\pmb{u}(t) \gets \tilde{\pmb{u}}(t)$ \COMMENT{keep candidate}
    \ENDIF \label{alg:online-validation:endcandsel}
    \STATE $u(t) \gets \pmb{u}_0(t)$ \COMMENT{apply input to system}
    \STATE $\tilde{\pmb{u}}(t+1) \gets \{\pmb{u}(t)_{1:N-1},K_\text{f}(\phi(N;x;\pmb{u}(t))\}$ \COMMENT{create candidate sequence} \label{alg:online-validation:appendterminal}
\ENDFOR
\end{algorithmic}
\end{algorithm}
\algo\ref{alg:online-validation} validates input sequences predicted by an approximate MPC \rev{$\Pi_\text{NN}$} and provides a safe fallback solution should the approximate solution be rejected.
For each newly measured state~$x(t)$, the NN evaluation gives a predicted input sequence~$\hat{\pmb{u}}(t)$, which is checked for feasibility in line~\ref{alg:online-validation:feasible}, \ie$\hat{\pmb{u}}(t)\in\mathcal{U}^{N}(x)$.
This requires the system dynamics to be integrated forward along the prediction horizon~(\ref{eqn:feasibletrajectories}).
If~$\hat{\pmb{u}}(t)$ is safe, we select the sequence of minimum cost~$\pmb{u}(t)$ from~$\hat{\pmb{u}}(t)$, and the safe candidate~$\tilde{\pmb{u}}(t)$ in line~\ref{alg:cost-decrease}.
Finally, the first input~$\pmb{u}_0(t)$ is applied to the system~(\ref{eqn:transition}) and a new safe candidate is created by shifting~$\pmb{u}(t)$ and appending the terminal controller.

For initialization of \algo\ref{alg:online-validation}, we state the following assumption:
\begin{assumption}\label{ass:initially-feasible}
The system is initialized in state~$x(0)$ with a safe candidate solution~$\pmb{u}_\text{\normalfont init}~\in~\mathcal{U}^{N}(x(0))$ in \algo\ref{alg:online-validation}.
\end{assumption}
\begin{remark}
Possible means of initialization could be: 
(i)~start in a steady state, where $\pmb{u}_\text{\normalfont init}$ is easy to construct;
(ii)~initialize with~$\pmb{u}_\text{\normalfont init}=\Pi_\text{\normalfont MPC}(x(0))$, which can be computed offline if the system is started in a known state~$x(0)$;
(iii)~initialize with~$\pmb{u}_\text{\normalfont init}=\Pi_\text{\normalfont NN}(x(0))$, which satisfies Assumption~\ref{ass:initially-feasible} if $\Pi_\text{\normalfont NN}(x(0))\in\mathcal{U}^N(x(0))$.
A practical alternative would be to use $\pmb{u}_\text{\normalfont init}=\Pi_\text{\normalfont NN}(x(0))$ even if~$\Pi_\text{\normalfont NN}(x(0))\notin\mathcal{U}^N(x(0))$ and always apply~$\Pi_\text{\normalfont NN}(x(t))$ in closed loop until $\Pi_\text{\normalfont NN}(x(t))\in\mathcal{U}^N(x(t))$, from which time on~\alg\ref{alg:online-validation} guarantees safety.
\end{remark}

The following theorem summarizes the closed-loop properties of \algo\ref{alg:online-validation}.
\begin{theorem}\label{theorem_main}
Let Assumptions~\ref{ass:mpcingredients} and~\ref{ass:initially-feasible} hold. Then, the closed-loop system resulting from \algo\ref{alg:online-validation} satisfies $(x(t), u(t)) \in \mathcal{X}\times\mathcal{U}$ for all~$t\in\mathbb{N}$ and converges to the state~$x=0$.
\end{theorem}

\begin{proof}
In the following, we first prove constraint satisfaction
and then convergence
for all $t\in\mathbb{N}$.

\textbf{Part I:}
Using induction, we show that~$\tilde{\pmb{u}}(t)\in \mathcal{U}^N(x(t))$ for all $t\in\mathbb{N}$.
Induction start: at $t=0$ by Assumption~\ref{ass:initially-feasible}.
Induction step: we know that~$\pmb{u}(t)\in\mathcal{U}^N(x(t))$, since~$\pmb{u}(t)$ is selected only from feasible sequences in lines~\ref{alg:online-validation:feasible} to \ref{alg:online-validation:endcandsel}.
Whenever~$\pmb{u}(t)\in\mathcal{U}^N(x(t))$, Assumption~\ref{ass:mpcingredients} ensures that the safe candidate~$\tilde{\pmb{u}}(t+1)$ constructed in line~\ref{alg:online-validation:appendterminal} is feasible, \ie$\tilde{\pmb{u}}(t+1)\in\mathcal{U}^N(x(t+1))$ (see standard candidate solution in nominal MPC proof~\cite{rawlings2017model}).
It follows by induction that $\pmb{u}\in\mathcal{U}^N(x(t)) \; \forall \; t\in\mathbb{N}$ and therefore $(x(t), u(t))\in\mathcal{X}\times\mathcal{U}$.

\textbf{Part II:}
For the shifted candidate sequence in \alg\ref{alg:online-validation}~line~\ref{alg:online-validation:appendterminal} it holds that
\begin{align}
\begin{split}
V(x(t&+1),\tilde{\pmb{u}}(t+1))\\
=&\sum_{k=1}^{N-1}\ell(\phi(k,x(t),\pmb{u}(t)),\pmb{u}_k(t))\\
&+\ell(\phi(N;x(t);\pmb{u}(t)),K_\text{f}(\phi(N;x(t);\pmb{u}(t)))\\
&+V_\text{f}(\phi(N;x(t+1);\tilde{\pmb{u}}(t+1))).
\end{split}
\end{align}
By Assumption~\ref{ass:mpcingredients}, it also holds that
\begin{align}
\begin{split}
V(&x(t+1),\tilde{\pmb{u}}(t+1))\\
=&V(x(t),\pmb{u}(t)) - \ell(x(t),\pmb{u}_0(t)) \\
&\hspace*{-0.7em}+\ell(\phi(N;x(t);\pmb{u}(t)),K_\text{f}(\phi(N;x(t);\pmb{u}(t))))\\
&\hspace*{-0.7em}-V_\text{f}(\phi(N;x(t);\pmb{u}(t)) + V_\text{f}(\phi(N;x(t+1);\tilde{\pmb{u}}(t+1)))\\
\stackrel{\mathclap{\text{\tiny by Ass.~\ref{ass:mpcingredients}}}}{\leq}& \hspace*{1.5em} V(x(t),\pmb{u}(t))-\ell(x(t),u(t)).
\end{split}
\end{align}
The second candidate~$\hat{\pmb{u}}(t+1)=\Pi_\text{NN}(x(t+1))$ has cost~$V(x(t+1),\hat{\pmb{u}}(t+1))$.
Because of line~\ref{alg:cost-decrease} in \algo\ref{alg:online-validation},~$V(x(t+1),\hat{\pmb{u}}(t+1))$ must be less than or equal to ~$V(x(t+1),\tilde{\pmb{u}}(t+1))$ to be selected:
\begin{align}
\begin{split}
V(&x(t+1),\pmb{u}(t+1)) \\
=&\min\left(V(x(t+1),\hat{\pmb{u}}(t+1)),V(x(t+1),\tilde{\pmb{u}}(t+1))\right)\\
\leq& V(x(t+1),\tilde{\pmb{u}}(t+1))\\
\leq& V(x(t),\pmb{u}(t))-\ell(x(t),u(t)).
\end{split}
\end{align}
By definition,~$V(x(0),\pmb{u}_\text{init})<\infty$.
Using the telescopic sum, it follows for the closed-loop system under \algo\ref{alg:online-validation} that
\begin{align}\label{eqn:cl-val-ineq}
\begin{split}
    \sum_{t=0}^{T-1} \alpha_\ell(&\norm{x(t)}) \\
    \stackrel{\mathclap{\text{by Ass.~\ref{ass:mpcingredients}}}}{\leq}&\hspace*{1.5em}\sum_{t=0}^{T-1}\ell(x(t),u(t)) \\
    \leq& \sum_{t=0}^{T-1} V(x(t),\pmb{u}(t))-V(x(t+1),\pmb{u}(t+1)) \\
    \stackrel{\mathclap{\text{tel. sum}}}{=}&\hspace*{1em}V(x(0),\pmb{u}_\text{init})-V(x(T),\pmb{u}(T))<\infty,
\end{split}
\end{align}
with~$V$ non-negative because~$\ell$ and~$V_\text{f}$ are positive definite.
The left-hand side of inequality~(\ref{eqn:cl-val-ineq}) is a series that is non-decreasing and bounded.
It therefore has to converge as~$T\rightarrow\infty$.
Hence,~$\lim_{t\rightarrow\infty}\norm{x(t)}=0$.
\end{proof}

\begin{remark}
The approximate policy~$\Pi_{\text{\normalfont NN}}$ together with \algo\ref{alg:online-validation} can be seen as a form of suboptimal MPC. \algo\ref{alg:online-validation} ensures that the chosen sequence reduces the cost function.
The guarantee for convergence can be strengthened to asymptotic stability by only applying the terminal controller or replacing the candidate with the terminal controller once the state reaches the terminal set, see \cite{allan2017inherent}. 
\end{remark}

Compared to naively applying the first input from~$\Pi_{\text{NN}}$ at every time step in closed loop, \algo\ref{alg:online-validation} results in additional computational cost due to a single forward simulation of the predicted input sequence (required by line~\ref{alg:online-validation:feasible}).
However, this computational cost for integrating system dynamics along the prediction horizon is considerably lower than online optimization of MPC or predictive safety filter problems~\cite{wabersich2021predictive}, which require multiple iterations, each involving at least one simulation of the system.
With our approach, we find a middle ground between full online computation (online optimization) and purely offline (naive NN) computation.
For only a minor increase in computational cost, we achieve \emph{deterministic} guarantees, in contrast to other approximate MPCs without guarantees~\cite{carius2020mpc,reske2021imitation} or probabilistic guarantees~\cite{hertneck2018learning,nubert2020safe,karg2021probabilistic}.
\rev{The presented result directly extends to reference tracking formulations~\cite{kohler2019nonlinear, nubert2020safe}, which we illustrate in our numerical examples.}

\section{Approximating robust MPCs}\label{sec:mpc:robust}
The proposed approximate~MPC with safety augmentation~(\sect\ref{sec:online-validation}) guarantees deterministic safety by applying a safe candidate if the NN approximation gives an input sequence that is infeasible for the~MPC problem~(\ref{eqn:nominalmpc}).
In this section, we provide an additional result deriving sufficient conditions for a safe NN approximation.
To this end, we first introduce a robust~MPC (Sec.~\ref{sec:rmpc}) and then obtain guarantees for approximating the robust MPC in (Sec.~\ref{sec:approx-rmpc}).
In the numerical examples~(Sec.~\ref{sec:numericalexamples}), we use this approach in a constructive way: by increasing the robustness to the MPC, the approximation yields a feasible trajectory more often.

\subsection{Robust MPC}\label{sec:rmpc}
We consider a perturbed input signal $\pmb{u}=\bar{\pmb{u}}+\pmb{d}_\text{w}$, with $\pmb{d}_\text{w}\in\mathcal{B}_\epsilon^N$ and $\mathcal{B}_\epsilon:=\{d_\text{w}\in\mathbb{R}^{n_\text{u}}\mid\norm{d_\text{w}}_\infty\leq\epsilon\}$, 
and a robust MPC scheme that guarantees stability and constraint satisfaction under such disturbances.
The robust MPC formulation is given by
\begin{align}\label{eqn:rmpc-optim}
    V^*(x) =&\min_{\bar{\pmb{u}}\in\bar{\mathcal{U}}^N(x)} V(x,\bar{\pmb{u}}),
\end{align}
where~$\bar{\mathcal{U}}^N(x) \subseteq \mathcal{U}^N(x)$ is the set of \textit{robustly} feasible input trajectories from a given state~$x$ resulting from the robustification.
The feasible set of~(\ref{eqn:rmpc-optim}) is~$\bar{\mathcal{X}}_\text{feas}$.
Solving this optimization problem for a current value~$x$ yields a corresponding sequence of inputs~$\bar{\pmb{u}}$, which we denote by
$\Pi_{\text{RMPC}}:\mathbb{R}^{n_\text{x}} \rightarrow \mathcal{U}^N$.
This sequence is also a suboptimal solution to (\ref{eqn:nominalmpc}).
The goal of the robustification is to find~$\bar{\mathcal{U}}^N(x)$, such that for any~$x\in\bar{\mathcal{X}}_\text{feas}$, it holds that
\begin{align}\label{ass:rmpc}
    \Pi_{\text{RMPC}}(x)\oplus \mathcal{B}_\epsilon^N\subseteq \mathcal{U}^N(x).
\end{align}

In the following, we briefly outline a simple constraint tightening method to satisfy~(\ref{ass:rmpc}), using Lipschitz continuity of the dynamics~\cite{manzano2020robust,marruedo2002input}.
To this end, we define the influence of the input disturbance on the state~$\tilde{\epsilon}:=\max_{(x,u)\in \mathcal{X}\times \mathcal{U}, d_{\text{w}}\in \mathcal{B}_\epsilon} \|f(x,u+d_\text{w})-f(x,u)\|_\infty$ and 
$\mathcal{B}_{\tilde{\epsilon}}:=\{\tilde{d}_\text{w}\in\mathbb{R}^{n_\text{x}}\mid\|\tilde{d}_\text{w}\|_\infty\leq\tilde{\epsilon}\}$.
The tightened constraints are
\begin{align}
\bar{\mathcal{U}} = \mathcal{U}\ominus \mathcal{B}_\epsilon, \quad \bar{\mathcal{X}}_k = \mathcal{X}\ominus c_k \mathcal{B}_{\tilde{\epsilon}}, \quad c_k = \sum_{i=0}^{k-1} L_{\text{f}}^i,
\end{align}
where $L_{\text{f}}$ is the Lipschitz constant of $f$ w.r.t. $x$, \ie $\|f(x_1,u)-f(x_2,u)\|_\infty\leq L_\text{f}\|x_1-x_2\|_\infty,\;\forall\;x_1,x_2\in\mathcal{X}, u\in\mathcal{U}$.
The feasible set in~(\ref{eqn:rmpc-optim}) is
\begin{align}
\begin{split}
\bar{\mathcal{U}}^N(x)=\{\bar{\pmb{u}}\in\bar{\mathcal{U}}^N\mid&\phi(k;x;\bar{\pmb{u}})\in\bar{\mathcal{X}}_k, k\in\mathbb{I}_{0:N-1},\\ &\phi(N;x;\bar{\pmb{u}})\in\bar{\mathcal{X}}_\text{f}\},
\end{split}
\end{align}
with an appropriate terminal set $\bar{\mathcal{X}}_\text{f}\subseteq\mathcal{X}_\text{f}$, see \cite[Thm.~1]{marruedo2002input}.
This Lipschitz-based method can be improved, \eg using a more general feedback and tube parametrization \cite{kohler2020computationally, sasfi2023robust}, which we also use in the examples in Sec.~\ref{sec:numericalexamples}.
Many other robust MPC methods to satisfy~(\ref{ass:rmpc}) also exist,~\eg\cite{houska2019robust}.

\subsection{Guarantees for approximations of robust MPC}\label{sec:approx-rmpc}
In this section, we use a robust~MPC (see Sec.~\ref{sec:rmpc}) to train an NN~$\Pi_\text{NN}$ using supervised learning as before.
Then, we provide conditions for the safety augmentation~(\alg\ref{alg:online-validation}) to not reject~$\Pi_\text{NN}$.
\begin{lemma}\label{lemma:epsilonapproximation}
Let~$\Pi_\text{\normalfont RMPC}$ satisfy~(\ref{ass:rmpc}) and suppose the approximation error of the NN satisfies~$\norm{\Pi_\text{\normalfont NN}(x)-\Pi_\text{\normalfont RMPC}(x)}_\infty\leq\epsilon \; \forall \; x\in\bar{\mathcal{X}}_\text{\normalfont feas}$.
Then, for all~$x\in\bar{\mathcal{X}}_\text{\normalfont feas}$:~$\Pi_\text{\normalfont NN}(x)\in\mathcal{U}^N(x)$.
\end{lemma}
\begin{proof}
Let~$\pmb{d}(x)=\Pi_\text{NN}(x)-\Pi_\text{RMPC}(x)$, then $\pmb{d}(x)\in\mathcal{B}_\epsilon^N \forall x\in\bar{\mathcal{X}}_\text{feas}$. As~$\Pi_{\text{RMPC}}$ satisfies~(\ref{ass:rmpc}), it follows that~$\Pi_\text{NN}(x) = \Pi_\text{RMPC}(x)+\pmb{d}(x)\in\bar{\mathcal{U}}^N(x)\subseteq\mathcal{U}^N(x)$, for which line~\ref{alg:online-validation:feasible} in \algo\ref{alg:online-validation} is always satisfied.
\end{proof}

\lem\ref{lemma:epsilonapproximation} allows us to design an approximate MPC that -- under \algo\ref{alg:online-validation} -- will often be applied to the system and therefore mitigate potentially long periods of simply replaying the safe candidate.
We emphasize that deterministic safety holds by \theo\ref{theorem_main} even if the conditions in \lem\ref{lemma:epsilonapproximation} do not hold.
The $\epsilon$ approximation of~$\Pi_\text{RMPC}$ is similar to prior work~\cite{hertneck2018learning, nubert2020safe}, where the~$\epsilon$ approximation is validated offline by performing simulations in closed loop with an NN controller, giving probabilistic safety guarantees directly.
In our setting, we do not require the uniform~$\epsilon$ approximation error bound everywhere, because safety is already guaranteed by \algo\ref{alg:online-validation}, making it considerably simpler to find an approximator for large systems.
In comparison to~\cite{hertneck2018learning,nubert2020safe}, \lem\ref{lemma:epsilonapproximation} serves as a sufficient condition, \eg on a (sub)set of sampled states.
We can use it in a constructive way to find ``good'' approximations~$\Pi_\text{NN}$ in our numerical examples.

\begin{remark}
\lem\ref{lemma:epsilonapproximation} does not guarantee that the proposed \algo\ref{alg:online-validation} will always apply the NN solution, because the candidate in line~\ref{alg:cost-decrease} could have preferable cost. However, it is safe to always apply the NN prediction in line~\ref{alg:cost-decrease}.
\end{remark}

\rev{\subsection{Robustness to bounded model-mismatch}
The previous subsection established that, if the approximation $\Pi_{\mathrm{NN}}$ reproduces a robust MPC solution to a uniform error of~$\epsilon$, then the sequence is guaranteed to lie in the nominal feasible set $\mathcal U^N(x)$.
In practice, however, additional robustness to bounded external disturbances~$\pmb{d}_\text{ext}$, thus an increased~$\epsilon$ beyond the approximation error $\|\pmb{d}_\text{NN}(x)\|$ might be desired.

\begin{lemma}\label{lemma:robustapproximation}
Let~$\Pi_\text{\normalfont RMPC}$ satisfy~(\ref{ass:rmpc}) and suppose the approximation error of the NN satisfies~$\norm{\Pi_\text{\normalfont NN}(x)-\Pi_\text{\normalfont RMPC}(x)}_\infty\leq\epsilon_\text{\normalfont{app}}<\epsilon$.
Then, for every external disturbances~$\pmb{d}_\text{\normalfont{ext}}(x)\in\mathcal{B}_{\epsilon_\text{\normalfont{ext}}}^N$ with~$\epsilon_\text{\normalfont{ext}}=\epsilon-\epsilon_\text{\normalfont{app}}$, we have $\Pi_\text{\normalfont{NN}}(x)+\pmb{d}_\text{\normalfont{ext}}(x)\in\mathcal{U}^N(x)$.
\end{lemma}
\begin{proof}
Let~$\pmb{d}_\text{app}(x)=\Pi_\text{NN}(x)-\Pi_\text{RMPC}(x)$.
Because $\epsilon_\text{ext} = \epsilon - \epsilon_\text{NN}$, we have $\mathcal{B}_{\epsilon_\text{app}}\oplus\mathcal{B}_{\epsilon_\text{ext}}=\mathcal{B}_{\epsilon}$.
Hence, $\Pi_\text{app}(x) + \pmb{d}_\text{ext}(x) = \Pi_\text{RMPC}(x) + \pmb{d}_\text{NN}(x) + \pmb{d}_\text{ext}(x) \in \Pi_{\text{RMPC}}(x)\oplus\mathcal{B}_\epsilon^N\subseteq\mathcal{U}_{N}(x)$. 
\end{proof}

We use Lemma~\ref{lemma:robustapproximation} by checking in Algorithm~\ref{alg:online-validation}, if the candidate sequence is in fact the solution to an $\epsilon_\text{ext}$ robust MPC problem.
Thus, the safety-augmented approximate MPC can treat bounded input disturbances.
}

\section{Numerical examples}\label{sec:numericalexamples}
In this section, we demonstrate our method on \rev{three} numerical example systems for non-linear MPC: a continuous stir tank reactor~\cite{mayne2011tube, hertneck2018learning}, a quadcopter model~\cite{kohler2020computationally,bouffard2012onboard}\rev{, and a robotic arm~\cite{nubert2020safe}}.
First, the system dynamics and MPC formulations are summarized.
The NNs are trained using supervised learning by generating optimal input sequences of a robust MPC design (Sec.~\ref{sec:mpc:robust}) over a large set of random initial states.
Finally, we compare our proposed approach, approximate MPC via safety-augmented NNs (Sec.~\ref{sec:online-validation}), with online optimization and naively (always) applying the NN approximation.
We benchmark the safety augmented NN with respect to computation time of online optimization and constraint satisfaction of the naive NN.
The code for all examples is available at~\url{https://github.com/hshose/soeampc.git}.

\subsection{Implementation details}
We use a terminal set~$\mathcal{X}_\text{f} = \{x\mid\|P_\text{f}^{0.5}x\|\leq\alpha\}$, terminal cost~$V_\text{f}=x^\top P_\text{f} x$, and terminal controller~$u=K_\text{f}x$, with parameters computed as in~\cite{kohler2019nonlinear}.
For simple robustification with respect to (hypothetical) input disturbances~$d_\text{w}$, we use a tube defined by~$\dot{s} = -\rho s + \bar{w}$ and~$s(0) = 0$, where~$\rho$ is a contraction rate and~$\bar{w}$ is a disturbance bound.
Here, $|d_\text{w}|$ is a design parameter that we use to facilitate learning a feasible solution.
The terminal set is tightened such that at time~$N$ the tube is inside the nominal terminal set.
The tube parameters~$\rho$ and~$\bar{w}$, stabilizing feedback~$K_\delta$, and constraint tightening factors~$c_j$ are computed following~\cite{kohler2020computationally,sasfi2023robust, nubert2020safe}.
The stabilizing feedback~$u=K_\delta x+v$, where $v$ is the optimization variable, is substituted directly into the MPC formulation.
The offline dataset generation \rev{for the stir tank and quadcopter} is performed using the acados toolbox~\cite{verschueren2022acados} with HPIPM as quadratic program (QP) solver, \rev{and for the robot arm using CasADi and IPOPT.}
\rev{We provide continuous time dynamics for our examples, which are discretized using acados implicit integrators for stir tank and quadcopter and explicit first order integration for the robot arm.}
\rev{We choose fully connected feedforward NNs such that the size of the approximator does not limit performance}.

\subsection{Benchmark systems}
In the following, we provide details of the dynamics and parameters of the benchmark systems.

\subsubsection{Stir tank reactor}
The non-linear stir tank reactor with~$x\in\mathbb{R}^2$ and~$u\in\mathbb{R}$ is used as a benchmark in previous work on MPC~\cite{mayne2011tube, hertneck2018learning}.
The continuous-time dynamics are
\begin{align}
    \dot{x} =
    \begin{bmatrix}
        \frac{1-x_1}{\theta}-kx_1\exp{(-\frac{M}{x_2})} \\
        \frac{x_\text{f}-x_2}{\theta} + k x_1\exp{(-\frac{M}{x_2})}- \gamma u(x_2-x_c))
    \end{bmatrix}
\end{align}
with parameters $\theta=20$, $k=300$, $M=5$, $x_\text{f}=0.3947$, $x_\text{c}=0.3816$, and~$\gamma=0.117$.
The goal is to drive the system to~$x_\text{e}=[0.2632, 0.6519]^\top$ and~$u_\text{e}~=~0.7853$.
The constraints are~$\|x-x_{\text{e}}\|_\infty\leq0.2$ and~$0\leq u\leq2$.
Initial states are sampled from the constraint set.
The sampling time is~$T_\text{s}=2$, and we use a horizon of~$N=10$ steps.
For the MPC, we choose quadratic cost with weights~$Q=I_2$ and~$R=10^{-5}$.
The terminal ingredients are computed as in~\cite{hertneck2018learning} with ~$|d_{w}|\leq1\cdot10^{-4}$.
We fit an NN with~\num{2} hidden layers,~\num{50} neurons per layer.

\subsubsection{Quadcopter}
We use a ten-state quadcopter model with an RMPC formulation as in~\cite{kohler2020computationally}.
We use the original system dynamics and parameters from~\cite{bouffard2012onboard}.
These dynamics are
\begin{align}\begin{split}
\dot{x}_i &= v_i, \, i\in\mathbb{I}_{1:3}, \quad 
\dot{v}_i = g \tan{(\phi_i)}, \, i\in\mathbb{I}_{1:2}, \\
\dot{v}_3 &= -g+ \frac{k_T}{m} u_3, \quad
\dot{\phi}_i=-d_1\phi_i+\omega_i, \, i\in\mathbb{I}_{1:2}, \\
\dot{\omega}_i&=-d_0\phi_i+n_0 u_i, \, i\in\mathbb{I}_{1:2},
\end{split}\end{align}
with state~$x=[x_1, x_2, x_3, v_1, v_2, v_3, \phi_1, \omega_1, \phi_2, \omega_2]^\top\in\mathbb{R}^{10}$ and input~$u=[u_1,u_2,u_3]^\top\in\mathbb{R}^3$.
The steady-state input is~$u_{\text{e},3}=\frac{gm}{k_\text{t}}$.
The constants are~$d_0=80$, $d_1=8$, $n_0=40$, $k_\text{T}=0.91$, $m=1.3$, and $g=9.81$.
Constraints are $x_1 \leq 0.145$ (a wall), $|\phi_{1,2}|\leq\frac{\pi}{9}$, $|u_{1,2}|\leq\frac{\pi}{4}$, and $u_3\in[0,2g]$.
The quadratic state and input cost terms have weights~$Q=\text{diag}([20\cdot1_3, 1_3, 0.01\cdot1_4])$ and~$R=\text{diag}([8,8,0.8])$.
Initial conditions are randomly sampled from
$|x_{1,2}|\leq2.5$, $x_1\leq0.145$, $|x_3|\leq3$, $|v_{1,2}|\leq3$, $|v_3|\leq5$,~$|\phi_{1,2}|\leq\frac{\pi}{9}$, and~$|\omega_{1,2}|\leq3\pi$.
The sampling time is~$T_\text{s}=0.1s$ with a horizon~$N=10$.
The terminal set and controller are computed by solving LMIs as described in detail in~\cite{kohler2019nonlinear,kohler2020computationally}.
We choose~$\|d_{w}\|_\infty\leq5\cdot10^{-3}$ and
fit an NN with 6 hidden layers and $[200,400,600,600,400,200]$ neurons.
\rev{
\subsubsection{Robot arm} We use a Franka robot arm\footnote{The robot has seven joints, we lock the third joint to avoid set-valued MPC solutions.} with~$q\in\mathbb{R}^6$, state~$x=[q^\top, \dot{q}^\top]^\top$, and action $u=\ddot{q}$. While the system dynamics are linear in joint space, the cartesian position and orientation of the end-effector depend on the forward kinematics,~$y=\text{FK}(x)$, which are highly nonlinear.
We impose box constraints on~$q$, $\dot{q}$, and $\ddot{q}$ according to datasheet limits and additionally avoid end-effector collisions with the environment via an interpolated signed-distance function constraint, $\text{SDF}(y)\geq d_\text{min}$.
Quadratic state, input, and output tracking cost terms have weights~$Q=\text{diag}([1_6, 10^{-4}\cdot1_6])$, $R=\text{diag}([10^{-4}\cdot1_6])$, and $Q_s=\text{diag}([100\cdot 1_7])$.
We choose a discretization time~$T_s=0.05s$ with a horizon~$N=20$.
Details on the MPC formulation and robustification are given in~\cite{nubert2020safe}.
We fit an NN with 9 hidden layers and 100 neurons per layer.
}

\begin{figure*}
    \begin{subfigure}[c]{0.4\textwidth}
    \begin{tikzpicture}[spy using outlines={circle, magnification=2,connect spies}]

    \definecolor{darkgray176}{RGB}{176,176,176}
    \definecolor{darkgreen}{RGB}{0,100,0}
    \definecolor{darkcyan}{RGB}{139,0,0}
    \definecolor{darkturquoise0191191}{RGB}{0,191,191}
    \definecolor{darkviolet1910191}{RGB}{191,0,191}
    \definecolor{goldenrod1911910}{RGB}{191,191,0}
    \definecolor{green01270}{RGB}{0,127,0}
    \definecolor{lightgray204}{RGB}{204,204,204}
    \definecolor{lightgray230}{RGB}{230,230,230}
    \definecolor{navy}{RGB}{0,0,128}
    
    \begin{groupplot}[group style={group size=1 by 2, vertical sep=0.4cm},
    every axis/.append style={
    font=\footnotesize
    },]
    \nextgroupplot[
    height=1.1in,
    legend cell align={left},
    legend style={
      fill opacity=1,
      draw opacity=1,
      text opacity=1,
      at={(1.05,0.5)},
      anchor=west,
      draw=lightgray204
    },
    tick align=outside,
    tick pos=left,
    width=2.8in,
    x grid style={lightgray230},
    xmajorgrids,
    xmin=0, xmax=1.7,
    xtick={0,0.5,1,1.5},
    xtick style={color=black},
    xticklabels=\empty,
    y grid style={lightgray230},
    ymajorgrids,
    ymin=-0.86393767082404, ymax=0.86393767082404,
    ytick style={color=black},
    ylabel={\(\displaystyle \textcolor{cyan}{u_1},\textcolor{darkgreen}{u_2} \)},
    ]
    \addplot [semithick, black, cyan, const plot mark right]
    table {%
    0 -0.159612354037842
    0.1 -0.159612354037842
    0.2 -0.647569589381573
    0.3 -0.636512983450999
    0.4 -0.101933264509976
    0.5 -0.752824279144703
    0.6 -0.785397882567309
    0.7 -0.785397882567309
    0.8 -0.785397882567309
    0.9 -0.777348718547889
    1 -0.370602343031474
    1.1 -0.264306616732982
    1.2 -0.200113072727079
    1.3 -0.117347502937007
    1.4 -0.0324578513812547
    1.5 0.0289649181709748
    1.6 0.061083656164881
    1.7 0.0724861966083692
    };
    \addplot [semithick, cyan, mark=*, mark size=2, mark options={solid}, only marks, forget plot]
    table {%
    0.6 -0.785397882567309
    0.7 -0.785397882567309
    0.8 -0.777348718547889
    };
    \addplot [semithick, green01270, const plot mark right]
    table {%
    0 0.785397882567309
    0.1 0.785397882567309
    0.2 0.785397882567309
    0.3 0.60325578373237
    0.4 -0.189979313009061
    0.5 -0.209220054396446
    0.6 -0.0707879100515836
    0.7 -0.00724493575070101
    0.8 -0.0359201981643486
    0.9 -0.0718982662270833
    1 -0.0973341937607215
    1.1 -0.0809703913348699
    1.2 -0.0638842475126794
    1.3 -0.0547441117759459
    1.4 -0.0495576361237228
    1.5 -0.04420407435846
    1.6 -0.0375772162541284
    1.7 -0.0307038445540006
    };
    \addplot [thick, darkgray, const plot mark right, densely dotted]
    table {%
    0 -0.159612354037842
    0.1 -0.159612354037842
    0.2 -0.647569589381573
    0.3 -0.636512983450999
    0.4 -0.101933264509976
    0.5 -0.752824279144703
    0.6 -0.785397882567309
    0.7 -0.785397882567309
    0.8 -0.780222929482828
    0.9 -0.737199890034967
    1 -0.487860520640103
    1.1 -0.30280095268667
    };
    \addplot [thick, darkgray, const plot mark right, densely dotted]
    table {%
    0 0.785397882567309
    0.1 0.785397882567309
    0.2 0.785397882567309
    0.3 0.60325578373237
    0.4 -0.189979313009061
    0.5 -0.209220054396446
    0.6 -0.0707879100515836
    0.7 -0.014873831093051
    0.8 -0.0402654372015852
    0.9 -0.0758336984914325
    1 -0.08731219078837
    1.1 -0.0773585646397204
    };
    \addplot [semithick, green01270, mark=*, mark size=2, mark options={solid}, only marks, forget plot]
    table {%
    0.6 -0.00724493575070101
    0.7 -0.0359201981643486
    0.8 -0.0718982662270833
    };
    \addplot [semithick, cyan, opacity=0.7, dashed, forget plot]
    table {%
    0 -0.785397882567309
    1.7 -0.785397882567309
    };
    \addplot [semithick, cyan, opacity=0.7, dashed, forget plot]
    table {%
    0 0.785397882567309
    1.7 0.785397882567309
    };
    \addplot [semithick, darkgreen, opacity=0.7, dashed, forget plot]
    table {%
    0 -0.785397882567309
    1.7 -0.785397882567309
    };
    \addplot [semithick, darkgreen, opacity=0.7, dashed, forget plot]
    table {%
    0 0.785397882567309
    1.7 0.785397882567309
    };
    
    \nextgroupplot[
    height=1.1in,
    legend cell align={left},
    legend style={
      fill opacity=1,
      draw opacity=1,
      text opacity=1,
      at={(1.05,0.5)},
      anchor=west,
      draw=lightgray204
    },
    tick align=outside,
    tick pos=left,
    width=2.8in,
    x grid style={lightgray230},
    xmajorgrids,
    xmin=0, xmax=1.7,
    xtick={0,0.5,1,1.5},
    xtick style={color=black},
    y grid style={lightgray230},
    ytick={-0.4, 0, 0.4},
    ymajorgrids,
    ymin=-0.47996558210444, ymax=0.47996558210444,
    ytick style={color=black},
    ylabel={\(\displaystyle \textcolor{cyan}{\phi_1},\textcolor{darkgreen}{\phi_2} \) [rad]},
    xlabel={time [s]},
    ]
    \addplot [semithick, cyan]
    table {%
    0 0.167308879527583
    0.1 0.203527948609371
    0.2 -0.08019287331007
    0.3 -0.360319257240699
    0.4 -0.43618012217706
    0.5 -0.418559832973062
    0.6 -0.397308630522791
    0.7 -0.385337378726638
    0.8 -0.37940241941461
    0.9 -0.356353335179994
    1 -0.274494986354963
    1.1 -0.166548282282944
    1.2 -0.0839558212553997
    1.3 -0.0336377173196301
    1.4 -0.00223759809672037
    1.5 0.0199970139119355
    1.6 0.0349802814949286
    1.7 0.0421347305680515
    };
    \addplot [semithick, green01270]
    table {%
    0 0.179447456311274
    0.1 -0.171006673214548
    0.2 -0.0590943773735548
    0.3 0.13949752812403
    0.4 0.16212963207673
    0.5 0.0505501536615187
    0.6 -0.0379084849684214
    0.7 -0.0581723893576722
    0.8 -0.0443389493323928
    0.9 -0.0320006111152253
    1 -0.0316409723679436
    1.1 -0.0352441459502089
    1.2 -0.0349122016533175
    1.3 -0.0307854126504436
    1.4 -0.0258249109875674
    1.5 -0.0217268925675225
    1.6 -0.0185063205148939
    1.7 -0.0156446796150189
    };
    \addplot [thick, black, densely dotted]
    table {%
    0 0.167308879527583
    0.1 0.203527948609371
    0.2 -0.08019287331007
    0.3 -0.360319257240699
    0.4 -0.43618012217706
    0.5 -0.418559832973062
    0.6 -0.397308630522791
    0.7 -0.384169625723359
    0.8 -0.369104824033479
    0.9 -0.339079128024287
    1 -0.275276989543668
    1.1 -0.185560976471916
    };
    \addplot [thick, black, densely dotted]
    table {%
    0 0.179447456311274
    0.1 -0.171006673214548
    0.2 -0.0590943773735548
    0.3 0.13949752812403
    0.4 0.16212963207673
    0.5 0.0505501536615187
    0.6 -0.0379084849684214
    0.7 -0.0590855572386659
    0.8 -0.0462306515679823
    0.9 -0.0341125086987617
    1 -0.0319013554120063
    1.1 -0.0335216416751173
    };
    \addplot [semithick, cyan, opacity=0.7, dashed, forget plot]
    table {%
    0 0.436332347367672
    1.7 0.436332347367672
    };
    \addplot [semithick, cyan, opacity=0.7, dashed, forget plot]
    table {%
    0 -0.436332347367672
    1.7 -0.436332347367672
    };
    \addplot [semithick, navy, opacity=0.7, dashed, forget plot]
    table {%
    0 0.436332347367672
    1.7 0.436332347367672
    };
    \addplot [semithick, navy, opacity=0.7, dashed, forget plot]
    table {%
    0 -0.436332347367672
    1.7 -0.436332347367672
    };
    \end{groupplot}
    
    \end{tikzpicture}
    
    \end{subfigure}\hspace{1em}
    \begin{subfigure}[c]{0.59\textwidth}
    \begin{tikzpicture}[spy using outlines={circle, magnification=4,connect spies}]

    \definecolor{darkgray176}{RGB}{176,176,176}
    \definecolor{green01270}{RGB}{0,127,0}
    \definecolor{lightgray204}{RGB}{204,204,204}
    
    \begin{axis}[
    every axis/.append style={
    font=\footnotesize
    },
    height=0.7*0.4*5*0.1818895740508146*6.8in,
    legend cell align={left},
    legend style={
      fill opacity=1,
      draw opacity=1,
      text opacity=1,
      at={(1.64,1)},
      anchor=north west,
      draw=lightgray204,
      align=left,
      row sep=5pt
    },
    tick align=outside,
    tick pos=left,
    width=0.7*6.8in*0.4,
    x grid style={darkgray176},
    xlabel={\(\displaystyle x_1\) [m]},
    xmin=0.2*-1.95025518960488, xmax=0.251494896068109,
    xtick style={color=black},
    y grid style={darkgray176},
    ylabel={\(\displaystyle x_2\) [m]},
    ymin=-0.381819069212184, ymax=0.0186563160372191,
    ytick style={color=black}
    ]
    \addplot [semithick, green01270]
    table {%
    -1.8501756402561 -0.363615642609939
    -1.54769992137782 -0.336402823032394
    -1.22760886597594 -0.322170604780663
    -0.915360814867084 -0.313120296252445
    -0.638121168523872 -0.291959318085371
    -0.405881422893838 -0.255870704931321
    -0.217256336444178 -0.214559490790164
    -0.0698946653229272 -0.176378325723328
    0.0375943595731836 -0.143634309188461
    0.106117710641977 -0.115267642624241
    0.138738513163606 -0.0901489884448692
    0.143903807945082 -0.0681648386736342
    0.132315015071025 -0.0496035005658727
    0.112190883155389 -0.0344362612896553
    0.0886126700628682 -0.0222837504924584
    0.0647437317029056 -0.0126728509336233
    0.0427783497046344 -0.00520084969823089
    0.0241808873572855 0.000452889434973518
    };
    \addlegendentry{safety augmented \\NN (ours)}
    \addplot [semithick, green01270, mark=*, mark size=3, mark options={solid}, only marks]
    table {%
    -0.217256336444178 -0.214559490790164
    -0.0698946653229272 -0.176378325723328
    0.0375943595731836 -0.143634309188461
    };
    \addlegendentry{safe candidate \\ solution}
    \addplot [thick, black, densely dotted]
    table {%
    -1.8501756402561 -0.363615642609939
    -1.54769992137782 -0.336402823032394
    -1.22760886597594 -0.322170604780663
    -0.915360814867084 -0.313120296252445
    -0.638121168523872 -0.291959318085371
    -0.405881422893838 -0.255870704931321
    -0.217256336444178 -0.214559490790164
    -0.0698919554367189 -0.176387737513979
    0.0378144513056821 -0.143742629980288
    0.107690780826373 -0.115654429225055
    0.143301108336366 -0.0910012719448919
    0.151415346719337 -0.0695090180561396
    };
    \addlegendentry{naive NN}
    \addplot[patch, patch type=rectangle, gray, solid, line width=0, opacity=0.5, forget plot] coordinates
    {
    (0.145, -1) (0.145, 1) (0.3, 1) (0.3, -1) 
    };
    \addplot [semithick, red, mark=x, mark size=2, mark options={solid}, only marks]
    table {%
    0.151415346719337 -0.0695090180561396
    };
    \addlegendentry{crash}

    \end{axis}

    \coordinate (spypoint) at (2.6,2.15);
    \coordinate (spyviewer) at (3.7,0.5);
    \spy [size=2.6cm] on (spypoint) in node [fill=white] at (spyviewer);
    
    \end{tikzpicture}
    
    \end{subfigure}
    \caption{Exemplary closed-loop simulation of naive and safety augmented NN.
    Naively (always) applying~$\Pi_\text{NN}$ (dotted line) for~\SI{1.1}{\second}, the quadcopter crashes ($\textcolor{red}\times$) into the wall at~$x_1=\SI{0.145}{\meter}$. With the safety-augmented NN presented in \algo\ref{alg:online-validation} (solid), a safe candidate is chosen over unsafe~$\Pi_\text{NN}$ predictions for three time steps~($\bullet$) and the system adheres to constraints.}
    \label{fig:quadcrash}
    \vspace{-1em}
\end{figure*}

\subsection{Simulation Results}
In the following, we present the results for the benchmark systems.
First, performance in terms of computation time is compared with online optimization.
Second, we show that our method is safe in \rev{all} examples, while naively (always) applying the NN leads to constraint violations in closed loop.
To gain some intuition into how well each approximator generalizes, $10^5$ points not seen during training are used to check if the NN provides a feasible solution to~(\ref{eqn:nominalmpc}).
The results are summarized in \tabe\ref{tab:dataoverview}.
In the stir tank and \rev{robot arm} example, the dataset is approximated well and more than~\SI{99}{\percent} of predictions are valid solutions to~(\ref{eqn:nominalmpc}) on test data.
For the quadcopter, an NN trained on a dataset with approximately~\num{5} gridpoints per state dimension yields a feasible solution in only~\SI{72}{\percent} of initial conditions.

\begin{table}
    \caption{Dataset and NN approximator overview. Number of points and cumulative offline computation time for the dataset generation, trainable parameters, and training time for the NN approximator. ``Solutions feasible'' is the percentage of test points, for which the NN gives a feasible solution to~(\ref{eqn:nominalmpc}).\looseness=-1}
    \label{tab:dataoverview}
    \begin{center}
    \begin{tabular}{p{1.3cm} P{0.6cm} P{1.5cm} P{1.0cm} P{1.0cm} P{1.0cm}}
    \toprule
    & \multicolumn{2}{c}{dataset} & \multicolumn{3}{c}{neural network} \\
    \cmidrule(lr){2-3}\cmidrule(lr){4-6}
    & points [$10^6$]
    & computation time\tablefootnote{\label{footnote:claix}in parallel on Intel Xeon Platinum 8160 "SkyLake" CPU at 2.1 GHz} [$10^3$~core-h]
    & parameters
    & training time\tablefootnote{\label{footnote:gpi}on a single NVIDIA A100 40GB GPU} [h]
    & solution feasible [\%]\tablefootnote{\label{footnote:teststates}independent samples from same set as training points, not seen in training}
    \\
    \midrule\midrule
    stir tank   & 0.9  & 0.5  & \num{3e3}  &  4  & 99  \\
    quadcopter  & 9.6  & 4.1  & \num{1e6}  &  37 & 72  \\
    robot arm  & 55.6 & 62.3 & \num{9e4} & 36  & 99 \\
    \bottomrule
    \end{tabular}
    \end{center}
    \vspace{-1em}
\end{table}

\subsubsection{Online computation times}
We evaluate the online computation times\footnote{\label{footnote:laptop}on AMD Ryzen 7 PRO 5850U "Zen3" mobile CPU at 4.4GHz turbo} for the proposed method and compare them with solving the NLP online.
The results of this comparison are summarized in Table~\ref{tab:computetimes}.
Computation times for our method include NN inference implemented in single-precision float using the C++ library Eigen, and forward simulation of the system dynamics using acados implicit fourth order integration \rev{for stir tank and quadcopter, and explicit first order integration for the robot arm.}
The combined online evaluation time is in the order of~\num{0.1} to \SI{0.2}{\milli\second} \rev{for implicit and~\SI{0.02}{\milli\second} for simpler integration.}
This is about four orders of magnitude faster than solving the NLP online and still at least ten times faster than a single solver iteration.
Yet, our method provides deterministic safety while a single SQP iteration might return an infeasible solution to~(\ref{eqn:nominalmpc}).

\begin{table}[bt]
    \vspace{-1em}
    \caption{Online computation times, comparing the proposed method implemented in C++ with online MPC solver acados~\cite{verschueren2022acados} \rev{(stir tank and quadcopter) and CasADi with IPOPT (robot arm)}. Our method is four orders of magnitude faster than solving (\ref{eqn:nominalmpc}) and even ten times faster than a single solver iteration, which is not sufficient to guarantee constraint satisfaction.}
        \label{tab:computetimes}
        \begin{center}
        \begin{tabular}{p{1.1cm} P{1.2cm} P{1.2cm} P{0.7cm} | P{0.45cm} P{0.45cm} | P{0.6cm}}
        \toprule
        & \multicolumn{3}{c}{\bfseries safety augmented NN (ours)}
        & \multicolumn{3}{c}{NLP solver} \\
        \cmidrule(lr){2-4}\cmidrule(lr){5-7}
        & \multirow{2}{1.2cm}{\centering NN inference [ms]}
        & \multirow{2}{1.2cm}{\centering forward simulation [ms]}
        & \multirow{2}{0.7cm}[-1em]{\centering \bfseries total [ms]}
        & \multicolumn{2}{P{1.2cm}}{\centering convergence [ms]}
        & \multirow{2}{0.6cm}{\centering single iter [ms]}
        \\
        \cmidrule(lr){5-6}
        & & & & mean & max & \\
        \midrule\midrule
        stir tank   & 0.001 & 0.082 & \textbf{0.083} & 430  & 1060 & 0.7 \\
        quadcopter  & 0.065 & 0.146 & \textbf{0.211} & 1592 & 9561 & 2.4 \\
        \midrule
        robot arm  & 0.011 & 0.008 & \textbf{0.019} & 600 & 4409 & 20.3 \\
        \bottomrule
        \end{tabular}
        \end{center}
        \vspace{-1em}
\end{table}

\subsubsection{Online constraint satisfaction (safety)}
To evaluate the effectiveness of safety augmentation (\alg\ref{alg:online-validation}), we compare the closed-loop system behavior with naively (always) applying the NN solution.
We randomly sample initial states
where the MPC has a feasible solution, \ie(\ref{eqn:nominalmpc}) is initially feasible.
We simulate the system dynamics and check constraint satisfaction in closed loop.
The results are summarized in \tabe\ref{tab:safevsnaive}.

\begin{table}[bt]
    \caption{Safety augmented vs. naive NN in closed loop. Starting in random initial conditions from the feasible set of the MPC, the naive NN results in constraint violations, while our proposed method is always safe.}
    \label{tab:safevsnaive}
    \begin{center}
    \begin{tabular}{p{1.1cm} P{2.15cm} P{1.28cm} | P{2.5cm} }
    \toprule
    & \multicolumn{2}{P{3.8cm}}{\textbf{[\%] of safe simulations}} 
    \\
    \cmidrule(lr){2-3}
    & naive (NN always)
    & \textbf{safe (ours)}
    & candidate applied [\%]
    \\
    \midrule\midrule
    stir tank   & 99.9 &  \textbf{100} & 48.8 \\
    quadcopter  & 84.8 &  \textbf{100} & 8.2 \\
    robot arm   & 93.5 &  \textbf{100} & 8.5 \\
    \bottomrule
    \end{tabular}
    \end{center}
    \vspace{-2em}
\end{table}

With the proposed safety augmented NNs, the systems adhere to constraints in closed loop -- they are \emph{deterministically} safe.
By contrast, naively applying the NN solution results in constraint violations for \rev{all} systems, underscoring the importance of safety augmentation.
Even when the NN provides an initially feasible solution to~(\ref{eqn:nominalmpc}), \ie $\Pi_\text{NN}(x(0))\in\mathcal{U}^N(x(0))$, it might not be safe to naively evaluate for all times in closed loop.
An illustrative example for catastrophic failure (crash into the wall) when naively applying the NN to the quadcopter system is depicted in Fig.~\ref{fig:quadcrash}: the NN controller is feasible for the first~\SI{0.6}{\second}.
Afterwards, the proposed safety augmentation detects unsafe NN input sequences and applies the safe candidate to safely evade the wall.
The naive implementation always applies the new NN input and crashes into the wall~\SI{0.3}{\second} later.
\rev{Similar behavior is observed for the robot arm in Fig.~\ref{fig:hardware}.}

\begin{table}[bt]
    \vspace{-1em}
    \caption{Reason for applying the safe candidate (can be multiple). In the stir tank reactor system, solutions are mostly rejected for having higher cost, while state constraint violations dominate in the quadcopter example.}
    \label{tab:rejectionreason}
    \begin{center}
    \begin{tabular}{p{1.6cm} P{1.9cm} P{2.4cm} P{1.0cm}}
    \toprule
    & {state constr. [\%] }
    & {terminal constr. [\%] }
    & {cost [\%] }\\
    \midrule\midrule
    stir tank   & 0.2 & 0.0 & \textbf{99.9} \\
    quadcopter  & \textbf{72.5} & 2.6 & 40.6 \\
    robot arm   & \textbf{99.8} & 24.7 & 95.8 \\
    \bottomrule
    \end{tabular}
    \end{center}
    \vspace{-1em}
\end{table}

\begin{table}[bt]
    \vspace{-1em}
    \caption{\rev{In-distribution robustness of the robot arm approximate MPC to external disturbances as percentage of feasible $\Pi_\text{NN}(x)$ over robustness level.}}
    \label{tab:robustness}
    \begin{center}
    \begin{tabular}{p{3cm} P{0.8cm} P{0.8cm} P{0.8cm} P{0.8cm}}
         \toprule
         \textbf{disturbance as frac. of $\bar{w}$}                 & 0\% & 50\% & 75\% & 95\%\\
        \midrule\midrule
         feasible in Alg.~\ref{alg:online-validation} & 99.7 & 98.1 & 94.9 & 84.5 \\
        \bottomrule
    \end{tabular}
    \end{center}
    \vspace{-1em}
\end{table}

\tabe\ref{tab:safevsnaive} also summarizes how often our method rejects the NN prediction and chooses the safe candidate solution in closed loop.
The reasons for rejection are given in \tabe\ref{tab:rejectionreason}, which could be state or terminal constraint violation (safety) in \alg\ref{alg:online-validation} line~\ref{alg:online-validation:feasible}, or failure to decrease the cost (convergence) in line~\ref{alg:cost-decrease}.
A trajectory could be rejected for multiple reasons.
The results in \tabe\ref{tab:rejectionreason} show that all three reasons for rejection might be relevant, depending on the specific application.

\begin{figure}
    \centering
    \begin{subfigure}[c]{0.48\linewidth}
    \centering
    \input{figure/hardware/rendering.tex}
    \end{subfigure}
    \begin{subfigure}[c]{0.48\linewidth}
    \centering
    \newlength{\pltwidth}
\newlength{\pltheight}
\setlength{\pltwidth}{1.8in}
\setlength{\pltheight}{1.0in}

\begin{tikzpicture}[every axis/.append style={
    font=\scriptsize,
    ticklabel style={font=\scriptsize},
    label style={font=\scriptsize},
    title style={font=\scriptsize},
    legend style={font=\scriptsize}
}]

\definecolor{green01270}{RGB}{0,127,0}
\definecolor{darkgray176}{RGB}{176,176,176}
\definecolor{lightgray204}{RGB}{204,204,204}

  \begin{axis}[
    name=top,
    width=\pltwidth,
    height=\pltheight,
    xlabel={},
    ylabel={\shortstack{distance to\\collision [m]}},
    ylabel style={xshift=-5pt},
    ymajorgrids,
    xtick={0,1,2,3,4,5},
    ytick={0,0.2},
    xticklabels={},
    xmin=0, xmax=4.5,
    ymin=-0.05, ymax=0.28,
    legend style={
        at={(0.4,1.1)},
        anchor=south,
        legend columns=2,
        font=\scriptsize,
        inner sep=2pt,
        column sep=2pt,
        draw=none
    },
    legend cell align=left,
    legend image code/.code={\draw[#1] (0pt,1pt) -- (10pt,1pt);}
  ]

    \foreach \x in {1.225,1.725,2.025,2.875,3.675,4.325} {
      \addplot+[
        line width=0.5mm,
        color=darkgray176,
        mark=none,
        forget plot
      ] coordinates {(\x,-1) (\x,1)};
    }

    \addplot[solid, semithick, dash pattern=on 0.5pt off 0.5pt, violet] table [x index=0, y index=1] {franka_cl_eval_mpc.dat};
    \addlegendentry{MPC (10 iters)}

    \addplot[solid, semithick, red] table [x index=0, y index=2] {franka_cl_eval.dat};
    \addlegendentry{naive NN}
    
    \addplot[solid, semithick, blue] table [x index=0, y index=3] {franka_cl_eval.dat};
    \addlegendentry{\makebox[0pt][l]{safety augmented NN (ours)}}

  \end{axis}

  \begin{axis}[
    name=bottom,
    at=(top.below south), anchor=above north,
    width=\pltwidth,
    height=\pltheight,
    xlabel={time [s]},
    xlabel style={yshift=5pt},
    ylabel style={xshift=-5pt},
    ymajorgrids,
    ylabel={\shortstack{distance to\\setpoint [m]}},
    xmin=0, xmax=4.5,
    xtick={0,1,2,3,4,5},
    ytick={0,0.4},
    ymin=-0.05, ymax=0.5,
  ]
        \foreach \x in {1.225,1.725,2.025,2.875,3.675,4.325} {
      \addplot+[
        line width=0.5mm,
        color=darkgray176,
        mark=none,
        forget plot
      ] coordinates {(\x,-1) (\x,1)};
    }

    \addplot[  solid, semithick, dash pattern=on 0.5pt off 0.5pt, violet] table [x index=0, y index=2] {franka_cl_eval_mpc.dat};

    \addplot[solid, semithick, red] table [x index=0, y index=8] {franka_cl_eval.dat};

    \addplot[solid, semithick, blue] table [x index=0, y index=9] {franka_cl_eval.dat};

  \end{axis}

    \draw[<-, red, very thick] (1.20,0.02) -- ++(-0.2,-0.2) node [anchor=east,text=red, xshift=2pt]{\scriptsize crash};
\end{tikzpicture}
    \vspace{-0.5em}
    \end{subfigure}
    \vspace{-1em}
    \caption{\rev{Franka Panda robot arm environment (left) and closed-loop simulation with velocity-proportional disturbances (right).
    Vertical lines indicate setpoint changes. The safety augmented NN (ours) shown in {\color{blue}blue} avoids collision while the naive NN in {\color{red}red} collides right before the 2s mark. The control performance of the NN controllers is close to the data-generating MPC with 10 solver iterations that take 200ms, which is too slow to run in real-time.}
    }
    \label{fig:hardware}
    \vspace{-1em}
\end{figure}
\rev{
\subsubsection{Robustness of the approximate MPC}
An evaluation of the robustness of the approximate MPC with safety augmentation is provided in Tab.~\ref{tab:robustness} for the robot arm setpoint tracking.
The percentages of input sequences by the NN that are feasible solutions to a~\emph{robust} MPC problem in Alg.~\ref{alg:online-validation} on random test-states\footref{footnote:teststates} are reported for increasing percentages of the disturbance level~$\bar{\omega}$; 100\% corresponds to the MPC used during training.
For example, when imposing robustness to external disturbances with~$75\%$ of $\bar{\omega}$, the safety augmentation  would use~$95\%$ of inputs by the NN -- a drop of only~$5\%$~pt. compared to the nominal case,~$0\%$ of $\bar{\omega}$.
As a result, both approximate (solid) and data-generating MPC with 10 solver iterations (dotted) show good performance even under strong disturbances in closed-loop (Fig.~\ref{fig:hardware}).
Notably, running the MPC solver with less iterations failed in this task.
}

\section{Conclusion}\label{sec:conclusion}
We propose approximate non-linear MPC with safety augmented~NNs, which provides deterministic safety and convergence in closed loop despite approximation errors.
We approximate entire input sequences and check online whether an input sequence by the NN is feasible and improves cost compared to a safe candidate sequence.
This candidate sequence is constructed from the shifted previous trajectory with appended terminal controller.
In \rev{three} numerical examples, we illustrate how our method provides safety and improves overall computation time by four orders of magnitude compared to solving NLPs online with acados.
\rev{During our numerical experiments, we encountered two major limitations: poor out-of-distribution generalization, which is a common issue in learning-based methods, and set-valued MPC solutions (e.g., from local minima), which significantly deteriorate performance with simple regression setups and necessitate generative models~\cite{julbe2025diffusion}.}
\rev{Further, safety augmentation could extend to other learning-based control methods (e.g., reinforcement learning), which requires additional work on these methods as they compute single actions, not action sequences.}

{\renewcommand{\baselinestretch}{0.963}\normalsize
\bibliographystyle{IEEEtran}
\bibliography{references_abbr}
}

\end{document}